\documentclass[%
aip,
jmp,
reprint, 
onecolumn 
]{revtex4-2}

\usepackage[utf8]{inputenc}
\usepackage[T1]{fontenc}

\usepackage{amsmath,amssymb,amsthm,anyfontsize,bm,dcolumn,diagbox,enumerate,enumitem,etoolbox,float,fullpage,graphicx,hhline,mathptmx,mathrsfs,mathtools,natbib,newtxmath,newtxtext,upref,xcolor,xurl}

\usepackage{hyperref}
\usepackage[capitalise]{cleveref} 

\setcitestyle{numbers,open={[},close={]}} 

\newtheorem{theorem}{Theorem}[section]

\newtheorem{cor}[theorem]{Corollary}
\newtheorem{lemma}[theorem]{Lemma}
\newtheorem{prop}[theorem]{Proposition}

\theoremstyle{definition} 
\newtheorem{definition}[theorem]{Definition}
\newtheorem{example}[theorem]{Example}

\theoremstyle{remark} 
\newtheorem{remark}[theorem]{Remark}

\numberwithin{equation}{section}

\DeclareMathOperator{\tr}{Tr}

\let\originalleft\left
\let\originalright\right
\renewcommand{\left}{\mathopen{}\mathclose\bgroup\originalleft}
\renewcommand{\right}{\aftergroup\egroup\originalright}

\newcommand{\bb}[1]{\mathbb{#1}}

\newcommand{\cl}[1]{\mathcal{#1}}

\newcommand{\inner}[2]{\left\langle {#1},{#2} \right\rangle}

\newcommand{\norm}[1]{\left\| #1 \right\|}

\makeatletter
\def\@email#1#2{%
 \endgroup
 \patchcmd{\titleblock@produce}
  {\frontmatter@RRAPformat}
  {\frontmatter@RRAPformat{\produce@RRAP{*#1\href{mailto:#2}{#2}}}\frontmatter@RRAPformat}
  {}{}
}%
\makeatother

\begin{document}


\title{Transitive Nonlocal Games}

\author{Prem Nigam Kar}
\affiliation{Section for Algorithms, Logic and Graphs, DTU Compute, Technical University of Denmark}

\author{Jitendra Prakash}
\affiliation{Department of Mathematics, University of New Orleans}

\author{David Roberson}
\affiliation{Section for Algorithms, Logic and Graphs, DTU Compute, Technical University of Denmark}
\affiliation{Center for Mathematics of Quantum Theory (QMATH), University of Copenhagen}

\date{\today}

\begin{abstract}
We study a class of nonlocal games, called transitive games, for which the set of perfect strategies forms a semigroup.
We establish several interesting correspondences of bisynchronous transitive games with the theory of compact quantum groups. In particular, we associate a quantum permutation group with each bisynchronous transitive game and vice versa. We prove that the existence of a $C^*$-strategy, the existence of a quantum commuting strategy, and the existence of a classical strategy are all equivalent for bisynchronous transitive games. We then use some of these correspondences to establish necessary and sufficient conditions for some classes of correlations, that arise as perfect strategies of transitive games, to be nonlocal.
\end{abstract}

\maketitle


\section{Introduction}\label{sec1}
Nonlocal games have been studied under several different names such as Bell scenarios and quantum pseudo-telepathy games across various disciplines. Their structure and formulation is simple, yet their applications are profound. In \cite{qiso1}, the authors introduced and studied the graph isomorphism game. The $(G,G)$-isomorphism game is also known as the $G$-automorphism game. The authors of \cite{qiso1} also proved that the set of all perfect quantum strategies for the graph automorphism game is closed under composition. Hence, these quantum strategies imitate the composability of graph automorphisms, which are precisely the perfect deterministic strategies for the $(G,G)$-isomorphism game. In \cite{lupini, soltan}, the authors also established several connections of nonlocal games with the theory of compact quantum groups and semigroups.

In \cite{wang}, Wang introduced quantum symmetry groups of finite spaces. However, this formulation was in terms of compact quantum groups. Hence, these quantum symmetry groups are not actual sets of transformations or functions like classical groups. In \cite{musto}, the authors introduced a notion of quantum bijections of finite classical and quantum sets. Quantum bijections of finite classical sets correspond to the perfect quantum strategies of the $K_n$-automorphism game as discussed in \cref{sec:transitive-games}. Since all perfect quantum strategies for bisynchronous games can be shown to be a subset of the set of perfect quantum strategies of this game, it makes sense to study nonlocal games with composable strategies, as they can be viewed as quantum analogues of classical permutation groups. 

Hence, we study \emph{transitive games}, which are nonlocal games with composable perfect strategies, and explore connections of bisynchronous transitive games with quantum permutation groups. These games were first studied in \cite{soltan}, and some connections of these games with compact quantum semigroups were established. However, this paper did not venture out to explore the perfect strategies for these games, which we do extensively in this paper. 

We prove that if a bisynchronous transitive game $\mathscr{G}$ has a perfect $C^*$-algebra strategy, then $C^*(\mathscr{G})$ is a quantum permutation group. Moreover, given any quantum permutation group $\mathbb{G}$ acting on $\{1,2,...,n\}$ we prove that there is a bisynchronous transitive game $\mathscr G$ such that $\mathbb{G}$ induces the same orbitals as $C^*(\mathscr{G})$. Moreover, $C^*(\mathscr{G})$ is the largest quantum permutation group inducing the same orbitals as $\mathbb{G}$ and every other quantum permutation group inducing the same orbitals as $\mathbb{G}$ can be obtained as a quotient of $C^*(\mathscr G)$. Hence, these transitive games can be treated as tools to define quantum permutation groups. Using the existence of the Haar state on a compact quantum group, we shall prove that the existence of a perfect $C^*$-strategy implies the existence of a quantum commuting strategy for bisynchronous transitive games, thus working towards  \cite[Problem 5.4]{paulsen}.
We shall then move on to studying coherent and transitive correlations, which are correlations that arise as perfect strategies of transitive games. We shall give necessary and sufficient conditions for these correlations to be nonlocal.

\subsection{Outline}
The article is organised as follows. In \cref{sec:prelims}, we begin with a brief introduction to nonlocal games, and compact quantum groups and semigroups. In \cref{sec:transitive-games}, we formally introduce transitive games and prove some preliminary results. In \cref{sec:trans-bisyn-games} we explore $C^*$-algebraic strategies for transitive games, where we establish several connections of these transitive games with quantum permutation groups, and finally in \cref{sec:coherent-transitive} we use these connections to give necessary and sufficient conditions for some families of correlations to be nonlocal.

\section{Preliminaries}\label{sec:prelims}
For an $n\in \bb N$, we let $[n]:=\{1,\dots,n\}$. For a nonempty finite set $W$, let $S_W$ be the \emph{symmetric group} on $W$. The identity element of the group $S_W$ is the identity function $\mathrm{id}_W:W\to W$.

The Hilbert spaces in this paper are over the complex field, and their inner products are linear in the first component and conjugate-linear in the second component. The algebra of all bounded linear operators from a Hilbert space $\cl H$ to itself is denoted by $\cl B(\cl H)$. A finite sequence $(P_i)_{i=1}^n$ of operators in $\cl B(\cl H)$ is called a \emph{projection-valued measure} (abbreviated \emph{PVM}) if each $P_i$ is a projection, that is, $P_i = P_i^* = P_i^2$, and their sum is the identity operator $I_{\cl H}$. A \emph{quantum state} is simply a unit vector in a Hilbert space.

All graphs considered in this paper are simple and without loops, unless stated otherwise. Let $G = (V(G), E(G))$ be a graph, where $V(G)$ and $E(G)$ denote the set of vertices and edges of $G$, respectively. For $x,y\in V(G)$, we write $x\sim y$ if $(x,y)\in E(G)$; $x\nsim y$ if $(x,y)\notin E(G)$; and $x\nsimeq y$ if $x$ and $y$ are distinct and nonadjacent.

\subsection{Nonlocal games and their strategies}\label{subsec:nonlocal-games}
A \emph{nonlocal game} $\mathscr G$ is a 6-tuple $\mathscr G = (X, Y, A, B, \pi, V)$, where $X, Y, A, B$ are nonempty finite sets,  $\pi: X\times Y\to [0,1]$ is a probability distribution, and $V: A \times B \times X \times Y \to \{0, 1\}$ is a function called the \emph{predicate}. Such a game $\mathscr G$ involves two players, Alice and Bob, and a referee. The sets $X$ and $Y$ are question sets for Alice and Bob, respectively, and the sets $A$ and $B$ are their corresponding answer sets. The probability vector $\pi$ determines the probability with which the referee sends each pair of questions $(x,y)\in X\times Y$ to the players. The function $V$ determines whether a pair of answers $(a,b)$ wins the game for Alice and Bob when given the questions $(x,y)$: they win if and only if $V(a,b,x,y) \equiv V(a,b|x,y) = 1$. The key property of nonlocal games is that the players cannot communicate once a round of the game begins.

Alice and Bob may employ different kinds of strategies based on the resources they have. The simplest kind of strategies are deterministic. A \emph{deterministic strategy} for Alice and Bob for the game $\mathscr G$ is a pair of functions $(f_A,f_B)$ where $f_A: X \to A$ and $f_B: Y \to B$. Upon receiving the questions $x$ and $y$, Alice and Bob return $a = f_A(x)$ and $b = f_B(y)$, respectively, to the referee. We say that the deterministic strategy $(f_A,f_B)$ is \emph{perfect} (or \emph{winning}) if $V(f_A(x),f_B(y)|x,y) = 1$ for all $(x,y)\in X\times Y$. Players may also use \emph{probabilistic} strategies utilizing \emph{local randomness} or \emph{shared randomness}. Such strategies are essentially convex combinations of deterministic strategies and, therefore, cannot yield better success rates. Deterministic strategies and their convex combinations are grouped together as \emph{local strategies} or \emph{classical strategies}. However, there may be better probabilistic strategies if the players utilise quantum resources. Let us first precisely define what we mean by a probabilistic strategy.

A \emph{probabilistic strategy} or a \emph{correlation} for the game $\mathscr G$ can be thought of as a collection of probability distributions $\{p(a,b|x,y): (a,b)\in A \times B\}$ for each $(x,y)\in X\times Y$. Here, $p(a,b|x,y)$ is the probability that Alice and Bob return answers $a\in A$ and $b\in B$, given that they received questions $x\in X$ and $y\in Y$, respectively. Such a correlation $\left(p(a,b|x,y)\right)_{a\in A, b\in B, x\in X, y\in Y}$ has some obvious properties: $p(a,b|x,y) \geq 0$ for all $a, b, x, y$, and $\sum_{a,b}p(a,b|x,y) = 1$ for all $x,y$. Moreover, a correlation must satisfy the following \emph{nonsignalling conditions}, stemming from the fact that the players cannot communicate during the game: \begin{align*}
 p_A(a|x) &:=  \sum_{b}p(a,b|x,y) = \sum_{b} p(a,b|x,y^{\prime}), \quad \text{ for all } a, x, y, y^{\prime}, \text{ and, } \\
 p_{B}(b|y) &:= \sum_{a}p(a,b|x,y) = \sum_{a} p(a,b|x^{\prime},y), \quad \text{ for all } b, y, x, x^{\prime}.
\end{align*} This general class of correlations, denoted by $C_{ns}(|X|,|Y|,|A|,|B|)$ is called the set of \emph{nonsignalling} correlations. (Here $|X|$ means the size of the set $X$.) For brevity's sake, we shall simply denote it by $C_{ns}$, and the same convention goes for other subsets of $C_{ns}$ that we will introduce next. 

Correlations arising via local strategies are called \emph{local} correlations, and the set of local correlations is denoted by $C_{loc}$.

In a \emph{quantum strategy} for the game $\mathscr G$, Alice and Bob share a quantum state $\psi \in \cl H_A\otimes \cl H_B$, where $\cl H_A$ and $\cl H_B$ are finite-dimensional Hilbert spaces and perform \emph{quantum measurements} on the shared quantum state. Such quantum measurements are given by PVMs: $(P_{x,a})_{a\in A}\subseteq \cl B(\cl H_A)$ corresponding to each question $x\in X$ for Alice, and $(Q_{y,b})_{b\in B} \subseteq \cl B(\cl H_B)$ corresponding to each question $y \in Y$ for Bob. Then, the quantum strategy $(\psi\in \cl H_A\otimes \cl H_B, \{P_{x,a}:x \in X,a \in A\}, \{Q_{y,b}:y \in Y,b \in B\})$ defines the following correlation: \begin{equation}\label{eqdef:q-corr}
p(a,b|x,y) = \inner{\left(P_{x,a}\otimes Q_{y,b}\right)\psi}{\psi}.
\end{equation} Correlations arising through quantum strategies are termed \emph{quantum} correlations, and the set of quantum correlations is denoted by $C_q$.

There is also a more general "quantum" strategy called the \emph{quantum commuting strategy} described as follows. Alice and Bob share a quantum state $\psi\in \cl H$, where $\cl H$ is a (possibly, infinite-dimensional) Hilbert space, quantum measurements for Alice given by PVMs $(P_{x,a})_{a\in A}\subseteq \cl B(\cl H)$ for each $x\in X$, quantum measurements for Bob given by PVMs $(Q_{y,b})_{b\in B}\subseteq \cl B(\cl H)$ for each $y\in Y$, with the condition that $P_{x,a}Q_{y,b} = Q_{y,b}P_{x,a}$ for all $a,b,x,y$. The corresponding correlation is: \begin{equation}\label{eqdef:qc-corr}
p(a,b|x,y) = \inner{P_{x,a}Q_{y,b}\psi}{\psi}.
\end{equation} Correlations arising through quantum commuting strategies are termed \emph{quantum commuting} correlations, and the set of quantum commuting correlations is denoted by $C_{qc}$. In general, quantum commuting strategies yield a bigger class of correlations than the quantum strategies \cite{Slofstra19,Slofstra20,DPP}. In fact, there exist quantum commuting strategies which \emph{cannot} even be approximated to an arbitrary degree by quantum strategies, that is, $\overline{C_q} \subsetneq C_{qc}$ \cite{Ji:2020apg}.

A correlation $(p(a,b|x,y))$ is called \emph{perfect} (or \emph{winning}) for the game $\mathscr G$ if $V(a,b|x,y) = 0$ implies $p(a,b|x,y) = 0$. In this paper we are only concerned with perfect correlations of a game $\mathscr G$, and therefore we will omit the probability distribution $\pi$ from the description of the game.

A game $\mathscr G = (X, Y, A, B, V)$ is called \emph{synchronous} if $X = Y =: I$, $A = B =: O$, and $V(a,b|x,x) = 0$ for all $x\in I$ and $a \neq b$. Additionally, $\mathscr G$ is called \emph{bisynchronous} if we have $V(a,a|x,y) = 0$ for all $a\in O$  and $x \neq y$. For brevity, we shall describe $\mathscr G$ simply as $\mathscr G = (I, O, V)$.

Graph-theoretic games make up most of the well-known examples of synchronous games. Consider the following two examples of synchronous games for graphs $G = (V(G),E(G))$ and $H = (V(H),E(H))$.

\emph{The $(G, H)$-homomorphism game.} Here, the question set is $I = V(G)$ and the answer set is $O = V(H)$. Alice and Bob need to convince the referee that they have a graph homomorphism from $G$ to $H$. The predicate $V$ is \begin{equation}
    V(a,b|x,y) = \begin{cases}
        0 &\text{ if } a \neq b \text{ and } x = y, \\
        0 &\text{ if } a \nsim b \text{ and } x \sim y, \\
        1 &\text{ otherwise.}
    \end{cases}
\end{equation} 

This game is synchronous, and it can be verified that it admits a perfect local strategy if and only if there is a graph homomorphism $\phi:G\to H$. A perfect quantum strategy (resp., quantum commuting strategy) for the $(G,H)$-homomorphism game is called a \emph{quantum homomorphism (resp., quantum commuting homomorphism) from $G$ to $H$}.

\emph{The $(G, H)$-isomorphism game.} Alice and Bob need to convince the referee that the graphs $G$ and $H$ are isomorphic. Thus, the predicate $V$ is $V(a,b|x,y) = 1$ if and only if $\mathrm{rel}(a,b) = \mathrm{rel}(x,y)$, where $\mathrm{rel}(x,y) \in \{=,\sim,\nsimeq\}$. This game is bisynchronous, and it can be verified that it admits a perfect local strategy if and only if $G$ is isomorphic to $H$. We say that $G$ is \emph{quantum isomorphic} (resp. \emph{quantum commuting isomorphic}) to $H$ if there exists a perfect quantum strategy (resp., quantum commuting strategy) for the $(G, H)$-isomorphism game. The $(G,G)$-isomorphism game will be called the $G$-automorphism game. The $G$-automorphism game always has a perfect deterministic strategy, e.g., the identity function.

Introduced in \cite{helton}, the \emph{$*$-algebra} $\mathcal{A}(\mathscr G)$ of a synchronous game $\mathscr G = (I, O, V)$ is defined to be the unital $*$-algebra generated by $\{u_{x,a}:x\in I, a \in O\}$ with the following relations: 
\begin{align}\label{alstatcon}
\begin{aligned}
u_{x,a} = u_{x,a}^* = u_{x,a}^2, &\quad\text{ for all } x\in I, a\in O, \\
\sum_{a\in O} u_{x,a} = 1, &\quad\text{ for all } x\in I, \\
u_{x,a}u_{y,b} = 0 &\quad\text{ if } V(a,b|x,y) = 0.
\end{aligned}\end{align} In general, $\mathcal{A}(\mathscr G)$ could be trivial. If $\mathcal{A}(\mathscr G)$ is nontrivial, the game $\mathscr G$ is said to have a \emph{perfect algebraic strategy}. Even if $\mathcal{A}(\mathscr G)$ is nontrivial, it may not be realizable as a $C^*$-algebra. (For a very brief review of $C^*$-algebras we refer the reader to \cref{subsec:CompactQuantumGroups}.) The game $\mathscr G$ has a \emph{perfect $C^*$-strategy} if there is a universal $C^*$-algebra with generators $\{u_{x,a}:x\in I, a \in O\}$ satisfying Relations~\ref{alstatcon}. We shall denote this universal $C^*$-algebra as $C^*(\mathscr G)$. The following theorem shows the relevance of the $*$-algebra $\cl A(\mathscr G)$ of a game $\mathscr G$.

\begin{theorem}[{\cite[Theorem 3.2]{helton}}]\label{t21}
Let $\mathscr G = (I, O, V)$ be a synchronous game. Then $\mathscr G$ has a perfect	\begin{enumerate}
\item[(a)] local strategy if and only if there exists a representation $\pi:\mathcal{A}(\mathscr G)\to \bb C$;
\item[(b)] quantum strategy if and only if there exists a representation $\pi:\mathcal{A}(\mathscr G) \to \cl B(\cl H)$ for some finite-dimensional Hilbert space $\cl H\neq 0$, in which case, the correlation $p$ induced by the perfect quantum strategy is given by $p(a,b|x,y) = \frac{1}{\dim{\mathcal{H}}}\tr(\pi(u_{x,a}u_{y,b}))$;
\item[(c)] quantum commuting strategy if and only if there exists a unital $*$-homomorphism $\pi: \mathcal{A}(\mathscr G) \to \cl B$ for some $C^*$-algebra $\cl B$ equipped with a faithful tracial state $\tau\in \cl B^*$, in which case, the correlation $p$ induced by the perfect quantum commuting strategy is given by $p(a,b|x,y) = \tau(\pi(u_{x,a}u_{y,b}))$, for all $a,b\in O, x,y\in I$.
\end{enumerate}
\end{theorem}

We turn our focus to correlations associated with question set $I$ and answer set $O$. Let $t\in \{loc,q,qc,ns\}$. A correlation $(p(a,b|x,y))\in C_t$ is called \emph{synchronous} if it satisfies $p(a,b|x,x) = 0$ for all $x \in I$ and $a \neq b.$ A \emph{bisynchronous correlation} is a synchronous correlation that satisfies $p(a,a|x,y) = 0$ for all $x \neq y$ and $a \in O$. We denote the subset of all synchronous (resp., bisynchronous) correlations in $C_t$ by $C^s_t$ (resp., $C^{bs}_t$). It is easy to see that perfect correlations of synchronous (resp., bisynchronous) games are synchronous (resp., bisynchronous).

\subsection{Compact quantum semigroups and groups}\label{subsec:CompactQuantumGroups}
We collect some definitions and results about compact quantum groups, mostly from  \cite[Chapter 1]{neshveyev}. We begin with a brief overview of $C^*$-algebras. A \emph{$C^*$-algebra} $\cl A$ is a Banach $*$-algebra whose norm satisfies $\norm{a^*a} = \norm{a}^2$ for all $a \in \cl A$. In this paper, all algebras considered are unital, and the unit is denoted by $1$. A linear functional $\varphi\in \cl A^*$ is called a \emph{state} if $\varphi(1) = 1$ and $\varphi(a^*a) \geq 0$ for all $a\in \cl A$. Additionally, $\varphi$ is \emph{faithful} if $\varphi(a^*a) = 0$ implies $a = 0$ for all $a \in \cl A$; and $\varphi$ is \emph{tracial} if $\varphi(ab) = \varphi(ba)$ for all $a,b \in \cl A$. A \emph{representation} of $\cl A$ is a unital $*$-homomorphism $\pi:\cl A\to \mathcal{B(H)}$ for some Hilbert space $\mathcal{H}$. The \emph{Gelfand-Naimark theorem} states that every $C^*$-algebra $\cl A$ is $*$-isomorphic to a norm-closed $*$-subalgebra of $\mathcal{B(H)}$ for some Hilbert space $\cl H$ \cite[Corollary II.6.4.10]{Blackadar-OA}.

Let $\cl G = \{g_i\}_{i \in \Omega}$ be a nonempty finite set of generators and $\cl R$ be a nonempty finite set of relations on $\cl G$. A \emph{representation} of $(\cl G|\cl R)$ is a set $\{T_i\}_{i \in \Omega}\subseteq \cl B(\cl H)$ (for some Hilbert space $\cl H$) satisfying the relations in $\cl R$, which also defines a representation of the free $*$-algebra $\langle \cl G \rangle$ generated by $\cl G$. For $a\in \langle \cl G\rangle$, let $\norm{a} := \sup \left\lbrace \norm{\pi(a)}: \pi \text{ is a representation of } (\cl G|\cl R) \right\rbrace.$ If $\max_{i\in \Omega}\norm{g_i}<\infty$, then $\norm{.}$ defines a seminorm on $\langle \cl G \rangle$. Quotienting out seminorm-zero elements and then taking completion, yields a $C^*$-algebra denoted by C$^*(\cl G|\cl R)$ and called the \emph{universal $C^*$-algebra generated by $(\cl G|\cl R)$}. A suitable reference on universal $C^*$-algebras is \cite[Section II.8.3]{Blackadar-OA}.

In this paper, a tensor product of two $C^*$-algebras is always meant to be their \emph{minimal} tensor product; see \cite[II.9.1.3]{Blackadar-OA} for a definition. The minimal tensor product is associative: $(\cl A\otimes \cl B) \otimes \cl C \simeq \cl A \otimes (\cl B \otimes \cl C)$ for $C^*$-algebras $\cl A, \cl B$ and $\cl C$ \cite[II.9.2.6]{Blackadar-OA}. Let $M_n(\cl A)$ be the set of all $n\times n$ matrices with entries from $\cl A$. We shall freely use the identification $\bb M_n \otimes \cl A \simeq M_n(\cl A)$, via, $\sum_{i,j=1}^n E_{i,j}\otimes a_{i,j} \leftrightarrow [a_{i,j}]_{i,j=1}^n$, where $\{E_{i,j}:i,j\in [n]\}$ is the set of canonical matrix units of $\bb M_n$, the space of all $n\times n$ complex matrices.

\begin{definition}
A pair, $\mathbb{G} = (\cl A,\Delta)$, where $\cl A$ is a $C^*$-algebra and $\Delta:\cl A\to \cl A\otimes \cl A$ is a unital $*$-homomorphism (called \emph{comultiplication}), is called a \emph{compact quantum semigroup} if the following holds:
\begin{enumerate}
\item[(a)] (\emph{Coassociativity}) $(\Delta \otimes \mathrm{id}) \Delta = (\mathrm{id} \otimes \Delta) \Delta$.
\end{enumerate} If the following \emph{cancellation property} also holds, we refer to $\mathbb{G}$ as a \emph{compact quantum group}.
\begin{enumerate}
\item[(b)] the spaces $(1 \otimes \cl A)\Delta(\cl A) := \mathrm{span}\{(1\otimes a)\Delta(b):a,b\in \cl A \}$ and  $(\cl A \otimes 1)\Delta(\cl A)$ are dense in $\cl A \otimes \cl A$.
\end{enumerate} For a compact quantum group $\mathbb{G} = (\cl A, \Delta)$ we shall write $C(\mathbb{G}) := \cl A$.
\end{definition}

Let $G$ be a compact topological group and $C(G)$ be the commutative $C^*$-algebra of continuous complex-valued functions on $G$. Since $C(G) \otimes C(G) \simeq C(G \times G)$, we may define $\Delta: C(G) \to C(G)\otimes C(G)$ by $\Delta(f)(g, h) = f(gh)$. Then, $(C(G), \Delta)$ is a compact quantum group. Moreover, all compact quantum groups $\mathbb{G}$ where the $C^*$-algebra $C(\mathbb G)$ is commutative, are of this type. Similarly, if $G$ is a compact topological semigroup, $C(G)$ is a compact quantum semigroup in a canonical manner, and all compact quantum semigroups $\mathbb{G} = (C(\mathbb{G}), \Delta)$ where $C(\mathbb{G})$ is commutative, are obtained in this way.

The \emph{convolution} of functionals $\omega_1, \omega_2 \in C(\bb G)^*$ is $\omega_1 * \omega_2 = (\omega_1 \otimes \omega_2)\Delta$. The following theorem states that every compact quantum group $\bb G$ admits a unique Haar state.

\begin{theorem}[{\cite[Theorem 1.2.1]{neshveyev}}]
For a compact quantum group $\mathbb{G}$, there exists a unique state $h\in C(\mathbb{G})^*$ (the \emph{Haar state}) such that for all $\omega \in C(\mathbb{G})^*$, we have $\omega*h = h*\omega = \omega(1)h$.
\end{theorem}

The following proposition gives a nontrivial example of a compact quantum group.

\begin{prop}[{\cite[Proposition 1.1.4]{neshveyev}}]\label{prop:gen-cqg-example}
Let $\cl A$ be the universal $C^*$-algebra generated by $\{u_{i,j}:i,j\in [n]\}$ with relations defined by the matrices $[u_{i,j}]_{i,j=1}^n$ and $[u_{i,j}^*]_{i,j=1}^n$ being invertible. Let $\Delta: \cl A \to \cl A \otimes \cl A$ be the unital $*$-homomorphism satisfying \begin{equation}\label{eqdef:comult-map}
    \Delta(u_{i,j}) = \sum_{k=1}^n u_{i,k}\otimes u_{k,j}, \qquad i,j\in [n].
\end{equation} Then, $(\cl A, \Delta)$ is a compact quantum group.
\end{prop}

As an application of \cref{prop:gen-cqg-example}, for a nonempty finite set $W$, we define the \emph{quantum permutation group} $S_W^+ = (\cl A, \Delta)$ as follows. Let $\cl A$ be the universal $C^*$-algebra generated by $\{u_{i,j}:i,j\in W\}$ and relations defined by the matrix $U = [u_{i,j}]_{i,j\in W}$ being a \emph{magic unitary}, that is, $u_{i,j} = u_{i,j}^* = u_{i,j}^2$ for each $i,j\in W$, and the matrix $U$ is a unitary $U^*U = UU^* = \mathrm{diag}(1,\dots,1)$. The comultiplication $\Delta$ is given by $\Delta(u_{i,j}) = \sum_{k\in W} u_{i,k}\otimes u_{k,j}$ for all $i,j\in W$. We usually denote $S_n^+ := S_{[n]}^+$.

For a compact quantum group $\mathbb{G}$, a matrix $U =  [u_{i,j}]_{i,j=1}^n \in M_n(C(\bb G))$ (for some $n\in \bb N$) is called a \emph{finite-dimensional unitary representation} (or concisely a \emph{representation}) if $U$ is unitary and Equation~\eqref{eqdef:comult-map} holds. Let $U$ and $V$ be representations of $\bb G$ on $\cl H_U$ and $\cl H_V$, respectively. An operator $T:\cl H_U \to \cl H_V$ is said to \emph{intertwine} $U$ and $V$ if $(T\otimes 1)U = V(T\otimes 1)$. The representations $U$ and $V$ are called \emph{unitarily equivalent} if the \emph{space of intertwiners}, $\mathrm{Mor}(U,V)$, contains a unitary element. A representation $U$ is called \emph{irreducible} if $\mathrm{Mor}(U,U) = \bb C$. We denote the set of representations of $\mathbb{G}$ by $\mathrm{Rep}(\mathbb{G})$, and $\mathrm{Irr}(\mathbb{G})$ will denote the set of irreducible ones. Every representation of $\mathbb{G}$ is a direct sum of irreducible ones \cite[Proposition 1.3.5 and Theorem 1.3.7]{neshveyev}. 

Let $U = \sum_{i,j=1}^n E_{i,j}\otimes u_{i,j}\in \bb M_n \otimes C(\bb G)$ be a representation of a compact quantum group $\bb G$. For each $\xi,\eta\in \bb C^n$, the elements \begin{equation*}
    (\omega_{\xi,\eta}\otimes \mathrm{id})U = \sum_{i,j=1}^n \inner{E_{i,j}\xi}{\eta} u_{i,j}
\end{equation*} are called the \emph{matrix coefficients} of $U$. The linear span $\bb C[\bb G]$ of matrix coefficients of all representations of $\bb G$, is a dense $*$-subalgebra of $C(\mathbb{G})$, which is invariant under comultiplication of $\mathbb{G}$. Hence, the pair $(\bb C[\bb G], \Delta)$ is a Hopf $*$-algebra \cite[Theorem 1.6.4]{neshveyev}. The definition of a Hopf $*$-algebra is given next. The linear maps $\epsilon$ and $S$ required in \cref{def:Hopf-*-algebra} exist and satisfy $(\mathrm{id}\otimes \epsilon) U = 1$ and $(\mathrm{id}\otimes S) U = U^{-1}$.

\begin{definition}\label{def:Hopf-*-algebra}
A pair, $(A,\Delta)$, where $A$ is a $*$-algebra and $\Delta: A \to A \otimes A$ is a unital $*$-homomorphism, is called a \emph{Hopf $*$-algebra} if: \begin{enumerate}
    \item[(a)] $(\Delta \otimes \mathrm{id}) \Delta = (\mathrm{id}\otimes \Delta) \Delta$,
    \item[(b)] there exists a linear map $\epsilon: A \to \bb C$ (called a \emph{counit}) such that $(\epsilon \otimes \mathrm{id}) \Delta(a) = (\mathrm{id} \otimes \epsilon) \Delta(a) = a,$ for all $a\in A$, and, 
    \item[(c)] there exists a linear map $S: A \to A$ (called an \emph{antipode}) such that $m(S\otimes \mathrm{id}) \Delta(a) = m(\mathrm{id} \otimes S) \Delta(a) = \epsilon(a)1,$ for all $a\in A$, and where $m:A\otimes A\to A$ is the multiplication map $a\otimes b \mapsto ab$.
\end{enumerate}
\end{definition}

From \cref{def:Hopf-*-algebra}, it follows that $\epsilon$ and $S$ are uniquely determined, $\epsilon$ is a unital $*$-homomorphism, $S$ is an anti-homomorphism, and $\epsilon S = \epsilon$.

For instance, the Hopf $*$-algebra $\cl O(S_W^+)$ of the compact quantum group $S_W^+$, is the universal $*$-algebra generated by $\{u_{i,j}:i,j\in W\}$ such that $u_{i,j} = u_{i,j}^* = u_{i,j}^2$ for all $i,j\in W$ and the matrix $[u_{i,j}]_{i,j\in W}$ is a magic unitary.

Let $X$ be a compact Hausdorff space and let $C(X)$ be the space of all continuous complex-valued functions on $X$. Through Gelfand transform, the correspondence $X \leftrightarrow C(X)$ is an equivalence between the category of compact Hausdorff spaces and continuous maps and the category of commutative unital $C^*$-algebras and unital $*$-homomorphisms \cite[Theorem II.2.2.6]{Blackadar-OA}. Motivated by this, the notion of a compact quantum space is given as a dual to the category of unital $C^*$-algebras. The action of a compact quantum group on a compact quantum space will be of interest to us:

\begin{definition}[{\cite[Section~2]{comer}}, {\cite[Definition 7]{soltan_quantum_2009}}]
A \emph{compact quantum space} $\bb X$ is identified with a unital $C^*$-algebra $C(\bb X)$. An \emph{action} of a compact quantum semigroup $\mathbb{G}$ on $\bb X$ is a unital $*$-homomorphism $\alpha : C(\bb X)\to C(\bb X) \otimes C(\mathbb{G})$ satisfying:

\begin{enumerate}
\item[(a)] $(\alpha \otimes \mathrm{id})\alpha = (\mathrm{id} \otimes \Delta) \alpha$
\end{enumerate}

If $\mathbb{G}$ is a compact quantum group, the \emph{action} of $\mathbb{G}$ on $\mathbb{X}$ is an action of $\mathbb{G}$ on $\mathbb{X}$ as a compact quantum semigroup satisfying the following additional property:
\begin{enumerate}
\item[(b)] $\mathrm{span}\{(1\otimes a)\alpha(f): a\in C(\bb G), f\in C(\bb X)\}$ is dense in $C(\bb X) \otimes C(\mathbb{G})$.
\end{enumerate}
\end{definition}

As an example, if $W$ is a nonempty set, the space $C(W)$ of all continuous complex-valued functions on $W$ may be expressed as the universal $C^*$-algebra generated by $\{e_i\}_{i\in W}$ and satisfying relations $e_i = e_i^* = e_i^2$ for all $i\in W$, $e_ie_j = 0$ if $i\neq j$, and $\sum_{i\in W} e_i = 1$. The quantum permutation group $S_W^+$ naturally acts on the space $C(W)$ via the left action: $\alpha(e_a) = \sum_{x\in W} e_x \otimes u_{x,a}$, for all $a\in W$.

\begin{definition}[\cite{banica05}]
Let $G = (V(G), E(G))$ be a graph with adjacency matrix $A_G$. The \emph{quantum automorphism group} of the graph $G$ is the compact quantum group $\mathrm{Qut}(G)$ where $C(\mathrm{Qut}(G))$ is the universal $C^*$-algebra generated by $\{u_{i,j}:i,j\in V(G)\}$ and with relations being that the matrix $U = [u_{i,j}]_{i,j \in V(G)}$ is a magic unitary and $A_G U = UA_G$, or equivalently that $U$ is a magic unitary and $u_{x,a}u_{y,b} = 0$ if. The comultiplication $\Delta$ is given by $\Delta(u_{i,j}) = \sum_{k\in V(G)}u_{i,k} \otimes u_{k,i}$ for all $i,j\in V(G)$.
\end{definition}

The quantum automorphism group $\mathrm{Qut}(G)$ of a graph $G$ is an example of a quantum permutation group acting on a set, which we define as follows. A \emph{quantum permutation group acting on a nonempty set $W$}, is compact quantum group $\bb G$ such that there exists a surjective unital $*$-homomorphism  $\pi$ from $C(S_W^+)$ onto $C(\bb G)$. In this case, the elements $\{\pi(u_{i,j}):i,j\in W\}$ are called the \emph{canonical generators} of $C(\bb G)$, and the matrix $[\pi(u_{i,j})]_{i,j\in W}$ is called the \emph{fundamental representation} of $\mathbb{G}$. The \emph{universal action} of $\bb G$ on $W$ is the action $u: C(W) \to C(W) \otimes C(\bb G)$ given by $e_i \to \sum_{j=1}^n e_j \otimes u_{j,i}$. Let $[v_{i,j}]_{i,j\in W}$ be the magic unitary generating $C(\bb G)$. In \cite{lupini}, the following equivalence relations were defined on $W$ and $W \times W$ respectively: \begin{enumerate}
	\item{$i \sim_1 j$ if $v_{i,j} \neq 0$;}
	\item{$(i,i^{\prime})\sim_2 (j,j^{\prime})$ if $v_{i,j}v_{i^{\prime},j^{\prime}}\neq0$}
\end{enumerate} The partitions of $W$ (resp., $W\times W$) with respect to $\sim_1$ (resp., $\sim_2$) are known as the \emph{orbits} (resp., \emph{orbitals}) of $\bb G$. The quantum orbitals of $\bb G$ form a coherent configuration \cite[Theorem 3.10]{lupini}, where a coherent configuration is defined as follows.

\begin{definition}\label{def:coherent-config}
A \emph{coherent configuration} on a finite nonempty set $W$ is a partition $\cl R = \{R_i: i\in \cl I\}$ of $W\times W$ where the classes $R_i$ satisfy the following properties: \begin{enumerate} 
\item[(a)] There is a subset $\cl J \subseteq \cl I$ such that $\{R_j:j\in \cl J\}$ is a partition of the diagonal $\{(x,x):x\in W\}$.
\item[(b)] For each $R_i$, its converse $R_i^{\mathrm{op}}:= \{(y,x):(x,y)\in R_i\}$ is again in $\cl R$.
\item[(c)] For all $i,j,k\in \cl I$ and $(x,z)\in R_k$, the number $p_{i,j}^k := |\{y\in W: (x,y)\in R_i \text{ and } (y,z)\in R_j\}|$ does not depend on $x$ and $z$.
\end{enumerate}
\end{definition}

Let $G$ be a permutation group acting on the set $W$. Then $W$ together with the partition $\cl R$ of $W\times W$ into orbits of $G$ (acting on $W\times W$ via $g(x, y) = (gx, gy))$ is a coherent configuration. A coherent configuration obtained in this way is called \emph{Schurian}.

\section{Transitive games}\label{sec:transitive-games}
We now introduce the class of nonlocal games of our interest. Any set $W$ appearing in this section is always assumed to be finite and nonempty.

\begin{definition}\label{def:transitive-game}
A game $\mathscr G = (X,Y,A,B,V)$ is called \emph{transitive} if $X = Y = A = B =: W$, and for each $(a, b, x, y) \in W^4$, if there exist $r,s\in W$ such that $V(a,b|r,s) = V(r,s|x,y) = 1$, then $V(a,b|x,y)=1$. Or equivalently (via contrapositive), whenever $V(a,b|x,y) = 0$ we have $V(a,b|r,s)V(r,s|x,y) = 0$ for all $r,s\in W$.
Such a transitive game $\mathscr G$ will be written $\mathscr G = (W, V)$.
\end{definition}

The following lemma is straightforward and another justification for the adjective ``transitive''.

\begin{lemma}\label{prop:game-from-relation}
Let $W$ be a set and let $\sim$ be a transitive relation over $W\times W$. Define a game $\mathscr G = (W,V)$ by $V(a,b|x,y) = 1$ if and only if $(a,b)\sim (x,y)$, for all $a,b,x,y\in W$. Then, $\mathscr G$ is transitive.

Conversely, let $\mathscr G = (W,V)$ be a transitive game. Then the relation $\sim$ over $W\times W$ defined by $(a,b)\sim (x,y)$ if and only if $V(a,b|x,y) = 1$, for all $a,b,x,y\in W$, is a transitive relation.
\end{lemma}

We first show that the set of perfect strategies of a transitive game, if nonempty, has a semigroup structure. Recall that a \emph{semigroup} is a set $S$ together with a function $.: S\times S\to S$ which satisfies the associative property: for all $x,y,z \in S$, we have $x.(y.z) = (x.y).z$. 

The following lemma is easily established.

\begin{lemma}[{\cite[Proposition 4.5]{ortiz}}, {\cite[Lemma 6.5]{PSSTW}}, {\cite[Proposition 3.4]{paulsen}}]\label{lem:product-corr-type} Let $C_{ns}$ be the set of nonsignalling correlations with identical question and answer sets $W$. For $t\in \{loc,q,qc,ns\}$, if $p,p^{\prime}\in C_t$, then the product $pp^{\prime}$ defined by: for $a,b,x,y\in W$, \begin{equation}\label{eqdef:product-correlation}
(pp^{\prime})(a,b|x,y) = \sum_{r,s \in W} p(a,b|r,s)p^{\prime}(r,s|x,y),
\end{equation} also belongs to $C_t$. Moreover, if $p,p^{\prime}$ are synchronous (or bisynchronous), then so is $pp^{\prime}$. The sets $C_t, C_t^s, C_t^{bs}$ are semigroups with respect to the product defined in Equation~\eqref{eqdef:product-correlation}.
\end{lemma}

\begin{prop}\label{prop:product-is-perfect}
Let $\mathscr G = \left(W, V\right)$ be a transitive game. Let $p,p^{\prime}\in C_{ns}$ be perfect correlations for $\mathscr G$. Then, the products $pp^{\prime}, p^{\prime}p\in C_{ns}$ are also perfect correlations for $\mathscr G$. In particular, for $t\in \{loc,q,qc\}$, the (possibly empty) subset $\mathscr C_t(\mathscr G) \subseteq C_t$ of all perfect $t$-strategies for the game $\mathscr G$ is a semigroup. Moreover, if $\mathscr G$ is synchronous (resp., bisynchronous), $\mathscr C_t(\mathscr G)$ is a sub-semigroup of $C^s_t$ (resp., $C^{bs}_t$).
\end{prop}

\begin{proof}
Let $V(a,b|x,y) = 0$ for some $a,b,x,y\in W$. By transitivity $V(a,b|r,s)V(r,s|x,y) = 0$ for all $r,s\in W$. That is, for all $r,s\in W$, we have either $V(a,b|r,s) = 0$ or $V(r,s|x,y) = 0$. Since $p$ and $p^{\prime}$ are perfect, we see that for all $r,s\in W$, we have either $p(a,b|r,s) = 0$ or $p^{\prime}(r,s|x,y) = 0$. In particular, the sum in Equation~\eqref{eqdef:product-correlation} is zero, so that $(pp^{\prime})(a,b|x,y) = 0$. Thus, $pp^{\prime}$ is a perfect correlation for $\mathscr G$. 

The remaining parts of the Proposition follow from \cref{lem:product-corr-type}. \end{proof}

Conversely, given a semigroup of correlations, the \emph{hardest} nonlocal game (as defined in \cite{LMPRSSTW}) that can be won using them is a transitive game. 

\begin{lemma}\label{l37}
    Let $\mathscr C\subseteq C_{ns}$ be a semigroup of correlations with identical questions and answers set $W$. Let $\mathscr G = (W,V)$ be the hardest game with $\mathscr C$ as its perfect strategies defined by $V(a,b|x,y) = 0$ if and only if $p(a,b|x,y) = 0$ for all $p\in \mathscr C$. Then $\mathscr G$ is transitive.
\end{lemma}

\begin{proof}
Suppose for some $a,b,x,y\in W$ we have $V(a,b|x,y) = 0$ and  $V(a,b|r,s) = V(r,s|x,y) = 1$ for some $r,s\in W$. Then, from the definition of the game, it follows that there exist $p, p^{\prime} \in \mathscr C$ such that $p(a,b|r,s) \neq 0$ and $p^{\prime}(r,s|x,y) \neq 0$. Using Equation~\eqref{eqdef:product-correlation}, we have that $(pp^{\prime})(a,b|x,y) \neq 0$. Since $\mathscr{C}$ is a semigroup, $pp^{\prime} \in \mathscr{C}$, and hence $V(a,b|x,y) \neq 0$, a contradiction. Hence, $\mathscr G$ is transitive.
\end{proof}

We now show that if we restrict our attention to bisynchronous transitive games then one can get better information about the set of perfect strategies of the games. To illustrate, we prove that the set of all perfect deterministic strategies of a transitive bisynchronous game, if nonempty, forms a group in a natural way.

\begin{prop}\label{prop:eq-conditions-perfect-strat}
	Let $\mathscr G = (W, V)$ be a transitive bisynchronous game. The following are equivalent.
	\begin{enumerate}
		\item[(a)] $\mathscr G$ has a perfect local strategy.
		\item[(b)] The set of all perfect deterministic strategies of $\mathscr G$ is a subgroup of the symmetric group $S_{W}$.
		\item[(c)] The identity map $\mathrm{id}_W$ is a perfect deterministic strategy.
	\end{enumerate}
\end{prop}

\begin{proof}
\emph{(a) implies (b).} If $\mathscr G$ has a perfect local strategy, we may as well take it to be deterministic. Because $\mathscr G$ is synchronous, any perfect deterministic strategy of $\mathscr G$ is a function $f:W\to W$, and it satisfies $V(f(x),f(y)|x,y) = 1$ for all $x,y\in W$. Suppose $f(x) = f(y)$; then since $\mathscr G$ is bisynchronous, we must have $x = y$. Hence, $f$ is injective, therefore, a permutation. Because $S_W$ is finite, it suffices to show that the set of perfect deterministic strategies is closed under composition. Let $f$ and $g$ be perfect deterministic strategies. Then for all $x,y\in W$, \begin{equation}\label{eq:f-g-perfect}
    V(f(x),f(y)|x,y) = 1, \qquad \text{ and }\qquad
    V(g(x),g(y)|x,y) = 1.
\end{equation} Since $g$ is bijective, we also have that for all $x,y\in W$, \begin{equation}\label{eq:fg-perfect}
    V(f\circ g (x), f\circ g(y)|g(x),g(y)) = 1.
\end{equation} Since $\mathscr G$ is transitive, using Equations \eqref{eq:f-g-perfect} and \eqref{eq:fg-perfect} we get that $V(f\circ g (x), f\circ g(y)|x,y) = 1$ for all $x, y \in W$. That is, $f\circ g$ is a perfect deterministic strategy, as required.

Finally, the implications, \emph{(b) implies (c)}, and \emph{(c) implies (a)}, are straightforward.
\end{proof}

For example, if $K_n$ is the complete graph on $[n]$, the $(K_n,K_n)$-homomorphism game is a transitive bisynchronous game, and the set of all its perfect deterministic strategies is the whole symmetric group $S_n$. On the other hand, consider the transitive game $\mathscr G = (W, V)$ defined by: $V(a,b|x,y) = \delta_{a,x}\delta_{b,y}$ for all $a,b,x,y\in W$. Clearly, the only perfect deterministic strategy is the identity function, and so the set of perfect deterministic strategies is a strict subgroup of $S_W$. Using \cref{prop:eq-conditions-perfect-strat}, we can construct examples of transitive bisynchronous games which do not admit perfect local strategies. For example, let $\mathscr G = (W,V)$ with $W = \{0,1\}$ and $V$ as in \cref{table:trans-game-no-perf-local}. Then $\mathscr G$ is a transitive bisynchronous game with no perfect deterministic strategies since $V(0,0|0,0) \neq 1$ and hence statement (c) of \cref{prop:eq-conditions-perfect-strat} is not fulfilled.

\begin{table}[H]
\centering
\begin{tabular}{|c|c|c|c|c|}
\hline
\diagbox{$(x,y)$}{$(a,b)$} & $(0,0)$ & $(0,1)$ & $(1,0)$ & $(1,1)$ \\
\hline
$(0,0)$ & 0 & 0 & 0 & 1 \\
\hline
$(0,1)$ & 0 & 1 & 1 & 0 \\
\hline
$(1,0)$ & 0 & 1 & 1 & 0 \\
\hline
$(1,1)$ & 0 & 0 & 0 & 1 \\
\hline
\end{tabular}
\caption{The predicate $V$ for a transitive bisynchronous game with no perfect local strategy.}
\label{table:trans-game-no-perf-local}
\end{table}

In general, graph homomorphism games are not transitive. In fact, if $G$ and $H$ are two graphs on the same set $W$ of vertices, then the $(G,H)$-homomorphism game is transitive if and only if $E(H) \subseteq E(H)$.

\section{Transitive bisynchronous games and quantum permutation groups}\label{sec:trans-bisyn-games}
We now focus on the $*$-algebra of a synchronous transitive game. Let $\mathscr G = (W, V)$ be a synchronous game with $\mathcal{A}(\mathscr G) \neq 0$. If $\mathscr G$ is transitive, then $\mathcal{A}(\mathscr G)$ admits a comultiplication $\Delta:\cl A(\mathscr G) \to \cl A(\mathscr G) \otimes \mathcal{A}(\mathscr G)$ defined by $\Delta(u_{x,a}) = \sum_{r\in W}u_{x,r}\otimes u_{r,a}$, for all $x,a\in W$, and also satisfies $(\Delta \otimes \mathrm{id}) \Delta = (\mathrm{id}\otimes \Delta) \Delta$ (see the proof of \cref{main-res} below). However, in general, $\cl A(\mathscr G)$ falls short of being a Hopf $*$-algebra, for if a counit $\epsilon:\cl A(\mathscr G) \to \bb C$ exists, then $V(x,y|x,y) = 1$ for all $x,y \in W$, which means that the game always has a perfect deterministic strategy. To prove the aforementioned claim, observe that a counit, if it exists, must satisfy $\epsilon (u_{x,a}) = \delta_{x,a}$ for all $x,a\in W$. Suppose for some $x,y\in W$, we have $V(x,y|x,y) = 0$. Then, $u_{x,x}u_{y,y} = 0$. Since $\epsilon$ is a homomorphism, $0 = \epsilon(u_{x,x}u_{y,y}) = \epsilon(u_{x,x})\epsilon(u_{y,y}) = \delta_{x,x}\delta_{y,y} = 1$, a contradiction.

Moreover, there is a natural $*$-homomorphism $\Phi:C(W)\to C(W) \otimes_{\text{alg}} \cl A(\mathscr G)$ defined by \begin{align}
    \Phi(e_a) = \sum_{x\in W} e_x \otimes u_{x,a}, \qquad \text{ for all } a\in W.
\end{align}

The following result was proved for synchronous transitive games in \cite{soltan}:

\begin{theorem}[{\cite[Theorem III.1]{soltan}}]
Let $\mathscr{G} = (W,V)$ be a synchronous transitive game with a perfect $C^*$-strategy, and let $C^*(\mathscr G)$ be the universal $C^*$-algebra generated by $\{u_{x,a}: x,a\in W\}$ satisfying Relations~\ref{alstatcon}. Then, there exists a unique $*$-homomorphism $\Delta: C^*(\mathscr{G}) \to C^*(\mathscr{G}) \otimes C^*(\mathscr{G})$ satisfying \begin{equation}
(\Phi \otimes \mathrm{id}) \circ \Phi = (\mathrm{id} \otimes \Delta) \circ \Phi
\end{equation} where $\Phi: C(W) \to C(W) \otimes C^*(\mathscr{G})$ is defined by $e_a \mapsto \sum_{x \in W} e_x \otimes u_{x,a}$ for all $a\in W$. Moreover, $\Delta$ is coassociative, and $\mathbb G = (C^*(\mathscr G), \Delta)$ is a compact quantum semigroup, and $\Phi$ defines an action of $\mathbb G$ on $W$.
\end{theorem}

The next theorem is a special case of {\cite[Theorem IV.1]{soltan}}. 

\begin{theorem}\label{main-res}
Let $\mathscr G = (W, V)$ be a bisynchronous transitive  game with a perfect $C^*$-strategy, and let $C^*(\mathscr G)$ be universal $C^*$-algebra generated by $\{u_{x,a}: x,a\in W\}$ satisfying Relations~\ref{alstatcon}. Then, $(C^*(\mathscr G),\Delta)$ is a compact quantum group with comultiplication defined by $\Delta(u_{x,a}) = \sum_{r\in W}u_{x,r}\otimes u_{r,a}$, for all $x,a\in W$. Moreover, it is a quantum permutation group acting on $W$.
\end{theorem}

\begin{proof} The universal $C^*$-algebra $C^*(\mathscr G)$ is generated by $\{u_{x,a}:x,a\in W\}$ satisfying Relations~\ref{alstatcon}. Let $v_{x,a} := \sum_{r\in W} u_{x,r} \otimes u_{r,a}$ for all $x,a\in W$. Using the fact that the game is synchronous it can be verified that $v_{x,a} = v_{x,a}^* = v_{x,a}^2$ for all $x,a\in W$, and that $\sum_{a} v_{x,a} = 1$.

Next, we note that \begin{equation*}
	v_{i,j}v_{k,l} = \Delta(u_{i,j})\Delta(u_{k,l})  = \left(\sum_{p}u_{i,p}\otimes u_{p,j}\right)\left(\sum_{q}u_{k,q} \otimes u_{q,l}\right)  = \sum_{p,q} u_{i,p}u_{k,q} \otimes u_{p,j}u_{q,l}
\end{equation*} We see that if $V(j, l | i, k) = 0$, then for all $k,l \in W$, we have that either $V(p,q|i, k) = 0$ or $V(j,l|p,q) = 0$. Hence, at least one of $u_{i,p}u_{k,q}$ and $u_{p,j}u_{q,l}$ is zero, which implies that $\Delta(u_{i,j})\Delta(u_{k,l}) = 0$. Thus, the elements $\{v_{x,a}:x,a\in W\}$ satisfy Relations~\ref{alstatcon}. By universal property, we get the existence of a unital $*$-homomorphism $\Delta:C^*(\mathscr G) \to C^*(\mathscr G)\otimes C^*(\mathscr G)$ satisfying $\Delta(u_{x,a}) = v_{x,a} = \sum_{r\in W}u_{x,r}\otimes u_{r,a}$, as required. Moreover, it is easy to check that $\Delta$ satisfies $(\Delta \otimes \mathrm{id}) \Delta = (\mathrm{id}\otimes \Delta) \Delta$. 

Next, let $U = [u_{x,a}]_{x,a\in W}$. Then, since all the entries of $U$ are self-adjoint projections, {\cite[Remark 2.1]{paulsen}} tells us that $\sum_{x}u_{x,a} = 1$ for all $x \in W$, so that $U$ is a magic unitary, that is, $UU^* = \mathrm{diag}(1,\dots,1) = U^*U$.

Hence, by \cref{prop:gen-cqg-example}, $(C^*(\mathscr G),\Delta)$ is indeed a compact quantum group, and since $U = [u_{x,a}]$ is a magic unitary, $C^*(\mathscr G)$ is a quantum permutation group acting on $W$. 
\end{proof}

Conversely, we establish that every quantum permutation group acting on a set gives rise to a transitive bisynchronous game. 

\begin{theorem}\label{thm:game-corr-to-qpg}
Let $\mathbb{G}$ be a quantum permutation group acting on a finite nonempty set $W$. Then, there exists a transitive bisynchronous game $\mathscr G = (W,V)$ with a perfect $C^*$-strategy such that the quantum orbitals of $(C^*(\mathscr G), \Delta)$ are the same as the quantum orbitals of $\mathbb{G}$.
\end{theorem}

\begin{proof}
Let $U = [u_{x,a}]_{x,a\in W}$ be the magic unitary generating $C(\mathbb{G})$. We define a game $\mathscr G = (W,V)$ with predicate given by $V(a,b|x,y) = 0$ if and only if $u_{x,a}u_{y,b} = 0$. We verify that $\mathscr G = (W,V)$ is a transitive bisynchronous game. Bisynchronicity is easy to see; we establish transitivity using a trick from \cite{lupini}. Let us assume that $V(a,b|c,d) = 1$ and $V(c,d|x,y) = 1$. That is, $u_{c,a}u_{d,b} \neq 0$ and $u_{x,c}u_{y,d} \neq 0$. Then, \begin{equation*}
    \Delta(u_{x,a}u_{y,b}) = \Delta(u_{x,a})\Delta(u_{y,b})  = \sum_{k,l \in W} u_{x,k}u_{y,l}\otimes u_{k,a}u_{l,b}.
\end{equation*} Multiplying by $u_{x,c}\otimes u_{c,a}$ from the left and $u_{y,d}\otimes u_{d,b}$ from the right, \begin{align*}
	(u_{x,c}\otimes u_{c,a})\Delta(u_{x,a}u_{y,b})(u_{y,d}\otimes u_{d,b}) = \sum_{k,l \in W} u_{x,c}u_{x,k}u_{y,l}u_{y,d}\otimes u_{c,a}u_{k,a}u_{l,b}u_{d,b} = u_{x,c}u_{y,d} \otimes u_{c,a}u_{d,b} \neq 0,
\end{align*} where we used $u_{x,c}u_{x,k} = \delta_{k,c}u_{x,c}$ (from synchronicity) and similar identities in the last equality. Hence, $V(a,b|x,y) = 1$, as needed for transitivity.

Finally, it is evident that the quantum orbitals of $C^*(\mathscr G)$ are the same as that of $\mathbb{G}$.\end{proof}

Owing to \cref{thm:game-corr-to-qpg}, given a quantum permutation group $\bb G$ acting on a finite nonempty set $W$ and generated by $\{u_{x,a}:x,a\in W\}$, we define the \emph{transitive bisynchronous game $\mathscr G = (W, V)$ associated with $\bb G$} by $V(a,b|x,y) = 0$ if and only if $u_{x,a}u_{y,b} = 0$. We should note that this correspondence is not one-to-one. Indeed, the game corresponding to a quantum permutation group only depends on its quantum orbitals. Hence, quantum permutation groups inducing the same orbitals have the same game corresponding to them. For example, $S_n$ and $S_n^+$ induce the same quantum orbitals, and the same corresponding game.

\begin{remark}
We note that $(C^*(\mathscr G), \Delta)$ is the ``largest'' compact quantum group which induces the same orbitals as $\mathbb{G}$ since $C^*(\mathscr G)$ is the universal $C^*$-algebra generated by $\{u_{x,a}\}_{x,a\in W}$ such that $U = [u_{x,a}]_{x,a\in W}$ is a magic unitary and $u_{x,a}u_{y,b} = 0$ if they are in different quantum orbitals in the coherent configuration induced by $C^*(\mathscr G)$. Hence, all quantum permutation groups acting on $W$ that induce the same orbitals as $\mathbb{G}$ can be written as quotients of $C^*(\mathscr G)$.   
\end{remark}

\begin{prop}
    Let $\bb G$ be a quantum permutation group acting on a finite nonempty set $W$ and generated by the magic unitary $U = [u_{x,a}]_{x,a\in W}$. Let $\mathcal{P}_{U}$ be the set of all quantum correlations given by \begin{equation*}
        p(a,b|x,y) = \rho(\pi(u_{x,a})\pi(u_{y,b}))
    \end{equation*} where $\pi: C(\mathbb{G}) \to \mathcal{A}$ is a $\ast$-homomorphism, and $\rho$ is a faithful tracial state on $\mathcal{A}$. Then, $\mathcal{P}_U$ is a semigroup. 
\end{prop}

\begin{proof}
Let $\pi_1:C(\mathbb{G})\to \mathcal{A}_1$ and $\pi_2:C(\mathbb{G})\to \mathcal{A}_2$ be two $*$-homomorphisms to $C^*$-algebras $\mathcal{A}_1, \mathcal{A}_2$ equipped with faithful tracial states $\rho_1$ and $\rho_2$ respectively. These two $*$-homomorphisms give rise to a new $*$-homomorphism $\pi: C(\mathbb G) \to \mathcal{A}_1 \otimes \mathcal{A}_2$ defined by $\pi = (\pi_1\otimes \pi_2)\circ \Delta$. Similarly, we may define a faithful tracial state on $\mathcal{A}_1 \otimes \mathcal{A}_2$ by setting $\rho = \rho_1 \otimes \rho_2$ (see \cite[Appendix]{avitzour1982free} for a proof of faithfulness of $\rho$). It is then easy to verify that if $p_1, p_2, p$ are the quantum correlations obtained from representations $\pi_1,\pi_2, \pi$ together with states $\rho_1, \rho_2, \rho$ respectively, then $p = p_2p_1$. Hence, $\cl P_U$ is indeed a semigroup. \end{proof}

\begin{remark}
We should note that this correspondence is not one-to-one either. For example, it follows from the results of \cite{roberson-schmidt} that $S_4$ and $S_4^+$, paired with their fundamental representations, give rise to the same semigroup of correlations.    
\end{remark}

\subsection{Equivalence of existence of local, quantum commuting, and \texorpdfstring{$C^*$}{C*}-algebraic strategies}
We now proceed to show that if a transitive bisynchronous game has a perfect $C^*$-strategy, then it has a perfect quantum commuting strategy.

We recall the following results from {\cite[Section 2]{bedos}}. Suppose $\bb G = (\cl A, \Delta)$ is a compact quantum group with Haar state $h$. Let $\cl O(\bb G)$ be the Hopf $*$-algebra of $\bb G$ equipped with counit $\epsilon$ and antipode $S$. In general, $h$ is faithful on $\cl O(\bb G)$ but not on the $C^*$-algebra $\cl A$. Let $N_h$ be the left kernel of $h$ which is a two-sided ideal of $\cl A$, and set $\cl A_r = \cl A/N_h$ to be the quotient $C^*$-algebra. Let $\theta:\cl A\to \cl A_r$ be the quotient map. Then, $\bb G_r := (\cl A_r, \Delta_r)$ is a compact quantum group where $\Delta$ is defined by $\Delta_r(\theta(a)) := (\theta \otimes \theta) \Delta(a)$ for all $a\in \cl A$. The Haar state $h_r$ on $\cl A_r$ satisfies $h = h_r \circ \theta$, and is faithful. The compact quantum group $\bb G_r$ is called the \emph{reduced compact quantum group} of $\bb G$ and $\theta$ is called the \emph{canonical map} from $\cl A$ onto $\cl A_r$. The Hopf $*$-algebra $\cl O(\bb G_r)$ is given by $\cl O(\bb G_r) := \theta(\cl O(\bb G))$, and its counit $\epsilon_r$ and antipode $S_r$ are given by $\epsilon = \epsilon_r\circ \theta$ and $\theta\circ S = S_r\circ \theta$, respectively. We say $\bb G$ is \emph{co-amenable} if the counit $\epsilon_r$ of $\cl O(\bb G_r)$ is norm-bounded (and hence can be extended to a $*$-homomorphism on $\cl A_r$). 

Let $\mathscr G = (W, V)$ be a transitive bisynchronous game with a perfect $C^*$-algebraic strategy, and let $(C^*(\mathscr G), \Delta)$ be the associated quantum permutation group acting on $W$. Then, the Haar state $h$ is tracial and faithful on $\mathcal{A}(\mathscr G)$. However, in general, the Haar state $h$ on a compact quantum group is neither faithful nor tracial on $C^*(\mathscr G)$; see \cite{bedos,woronowicz2} for examples of such compact quantum groups. It is known that even for quantum permutation groups, the Haar state may not be faithful, see \cite{banica99} or \cite[Section 2.7]{neshveyev}.

\begin{theorem}[{\cite[Theorem 2.2]{bedos}}]
A compact quantum group $\bb G$ is co-amenable if, and only if, its Haar state is faithful and its counit is norm-bounded.
\end{theorem}

\begin{theorem}[{\cite[Proposition 1.7.9]{neshveyev}}] Let $\bb G$ be a compact quantum group and $h$ be its Haar state. Let $\cl{O}(\bb G)$ be the Hopf $*$-algebra of $\bb G$, and $\epsilon$ and $S$ be the counit and antipode of $\cl{O}(\bb G)$, respectively. Then, $h$ is tracial if, and only if, $S^2 = \mathrm{id}$. (Such a compact quantum group is said to be of \emph{Kac type}.)
\end{theorem}

\begin{prop}\label{prop:Haar-on-qpg}
The Haar state $h_r$ on a reduced quantum permutation group is both faithful and tracial.
\end{prop}

\begin{proof}
Let $\bb G$ be a quantum permutation group. Then, it is of Kac type (see {\cite[Example 1.7.10]{neshveyev}}) which implies that its Haar state is tracial. Hence, the Haar state $h_r$ on the reduced quantum group satisfies: for all $a,b\in \bb G_r$, $h(\theta(ab)) = h(\theta(a)\theta(b)) = h(\theta(b)\theta(a)) = h(\theta(ba)).$ Hence, $h_r$ is tracial. Faithfulness of $h_r$ is by the construction described in the previous paragraph.
\end{proof}

Now we can prove our result about the existence of a perfect qc strategy.

\begin{theorem}
	Let $\mathscr G = (W, V)$ be a transitive bisynchronous game. If $\mathscr G$ has a perfect $C^*$-strategy, it has a perfect quantum commuting strategy.
\end{theorem}

\begin{proof}
Let $\bb G = (C^*(\mathscr G),\Delta)$ denote the quantum permutation group associated with the game generated by $\{u_{x,a}:x,a\in W\}$ and satisfying Relations~\ref{alstatcon}. Let $\bb G_r = (C^*(\mathscr G)_r,\Delta_r)$ be the reduced compact quantum group of $\bb G$. By \cref{prop:Haar-on-qpg} the Haar state $h_r$ on $C^*(\mathscr G)_r$ is faithful and tracial. Consider the projections $\{\theta(u_{x,a}):x,a\in W\}\subseteq C^*(\mathscr G)_r$. It is clear that $U = [\theta(u_{x,a})]_{x,a\in W}$ is a magic unitary. Moreover, if $V(a,b|x,y)=0$, then $h_r(\theta(u_{x,a}u_{y,b})) = h(u_{x,a}u_{y,b}) = 0$, where the last equality follows from $u_{x,a}u_{y,b} = 0$; so that $\theta(u_{x,a})\theta(u_{y,b}) = 0$, as needed. Hence, $\mathscr G$ has a perfect quantum commuting strategy by \cref{t21}.
\end{proof}

Now, we shall show that existence of a perfect qc-strategy (or equivalently, existence of a perfect $C^*$-strategy) implies the existence of a perfect classical strategy for transitive bisynchronous games.

\begin{prop}
    Let $\mathscr{G} = (W, V)$ be a bisynchronous transitive game with a perfect quantum commuting strategy. Then, $V(x,y|x,y) = 1$ for all $x,y \in W$, which implies that $\mathscr{G}$ has a perfect classical strategy.
\end{prop}

\begin{proof}
    If $\mathscr{G}$ has a perfect quantum commuting strategy, $C^*(\mathscr{G})$ is non-trivial. Hence, if $[\theta(u_{x,a})]$ and $h_r$ are the magic unitary and the Haar state of the reduced quantum permutation group, respectively, then $p(a,b|x,y) = h_r(\theta(u_{x,a}u_{y,b}) = h(u_{x,a}u_{y,b})$ is a perfect quantum commuting strategy for $\mathscr{G}$. In particular, it follows from \cite[Theorem 4.2]{lupini} that for any $x,y \in W$, we have that $p(x,y|x,y) = h(u_{x,x}u_{y,y}) = \frac{1}{|\Omega_i|} \neq 0$, where $\Omega_i$ is the quantum orbital of $C^*(\mathscr{G})$ containing $(x,y)$. Since $p(x,y|x,y) \neq 0$ for a perfect strategy, we see that $V(x,y|x,y) = 1$. Hence, $\mathscr{G}$ has a perfect classical strategy.
\end{proof}

\subsection{Quantum permutation groups arising from transitive nonlocal games}
We end this section by showing that transitive nonlocal games satisfying $V(x,y|x,y) = 1$ for all $x,y \in W$ directly arise from graph automorphism games.

\begin{prop}\label{graph-game-cond}
    Let $V$ be the predicate of a transitive synchronous game. Then, $V$ is the predicate of the graph automorphism game of an complete edge labelled directed graph with loops if and only if satisfies the following conditions: \begin{enumerate}
        \item[(a)] $V(a,b|a,b) = 1$ for all $a,b \in W$
        \item[(b)] $V(a,b|x,y) = V(x,y|a,b)$ for all $a,b,x,y \in W$.
    \end{enumerate}
\end{prop}

\begin{proof}
Since the the predicate of the game is transitive, and it satisfies (a) and (b), it is easy to see that the relation $\sim$ defined in \cref{prop:game-from-relation} becomes an equivalence relation. Hence, we may divide $W \times W$ into equivalence classes $E_1, E_2, ..., E_m$. Let us colour the complete directed graph (including loops) with vertex set $W$ according to these partitions, i.e $(a,b)$ is coloured with the colour $i \in \{1,2,...,m\}$ if and only if it is in the equivalence class $E_i$. Let us call this graph $K_{(W,V)}$. It is easy to see that $V(a,b|x,y) = 1$ if and only if $(a,b) \sim (x,y)$, which happens if and only $(a,b)$ and $(x,y)$ have the same colour. Hence, $V$ is indeed the predicate of the automorphism game of the edge coloured complete directed graph $K_{(W,V)}$. The other direction is quite easy to see.
\end{proof}

\begin{prop}
    Let $\mathscr{G}$ be a bisynchronous transitive game with a perfect quantum strategy. Then, $(C^*(\mathscr{G}), \Delta)$ is the quantum automorphism group of a complete edge and vertex labelled directed graph,
\end{prop}
\begin{proof}
    Let $\mathscr{G}$ be a bisynchronous transitive game with a perfect quantum strategy, and let $\widetilde{\mathscr{G}}$ be the nonlocal game associated with the quantum permutation group $(C^*(\mathscr{G}), \Delta)$. Then, it is clear that $C^*(\widetilde{\mathscr{G}}) = C^*(\mathscr{G})$. 

    Since the predicate $V_{\widetilde{\mathscr{G}}}$ of $\widetilde{\mathscr{G}}$ is given by $V(a,b|x,y) = 1$ if and only if $(x,a)$ and $(y,b)$ are in the same orbital of $C^*(\mathscr{G})$, we see that it satisfies the hypothesis of \cref{graph-game-cond}. Hence, $\widetilde{\mathscr{G}}$ is the predicate of a graph automorphism game of a complete edge labelled directed graph with loops, say $\Gamma$.

    Now, we shall turn $\Gamma$ into a simple complete edge and vertex labelled directed graph $\widetilde{\Gamma}$ by colouring removing the loops, and colouring the vertices in the same colour as the removed loops. We claim that $V_{\widetilde{\mathscr{G}}}(a,b|x,y) = 1$ if and only if \begin{enumerate}[label = \roman*.]
        \item the vertices $x,a$ have the same colour in $\widetilde{\Gamma}$,
        \item the vertices $y,b$ have the same colour $\widetilde{\Gamma}$,
        \item and the edges $(x,a), (y,b)$ have the same colour $\widetilde{\Gamma}$.
    \end{enumerate}
    
    Let $u$ be the fundamental representation of $C^*(\mathscr{G})$. Then, $V_{\widetilde{\mathscr{G}}}(a,b|x,y) = 1$, if and only if $u_{x,a}u_{y,b}\neq 0$. By the construction in \cref{thm:game-corr-to-qpg}, we see that the edges $(x,a), (y,b)$ have the same colour in $\Gamma$, and hence in $\widetilde{\Gamma}$. Similarly, since $u_{x,a}u_{y,b} \neq 0$, we have that $u_{x,a} \neq 0 \Rightarrow u_{x.a}u_{x,a} \neq 0$, and $u_{y,b} \neq 0 \Rightarrow u_{y,b}u_{y,b} \neq 0$. Hence, the loops $(x,x), (a,a)$ and $(y,y), (b,b)$  have the same colour in $\Gamma$, which implies that the vertices $x,a$ and the vertices $y,b$ have the same colour in $\widetilde{\Gamma}$. The other direction of the implication can be proved similarly. 

    Hence, $V_{\widetilde{\mathscr{G}}}$ is the predicate of the graph automorphism game of the graph $\widetilde{\Gamma}$. Hence, $(C^*(\mathscr{G}), \Delta) \cong (C^*(\widetilde{\mathscr{G})}, \Delta) \cong \text{Qut}(\widetilde{\Gamma})$.
\end{proof}

\begin{remark}
    The above proposition implies that not all quantum permutation groups can be realised as the C*-algebra of a bisynchronous transitive game. Indeed, following the same proof as \cite[Theorem 3.7]{banica_mccarthy_2022}, we see that $S_5^+$ has uncountably many subgroups. However, there are only countably many complete edge and vertex labelled graphs. Hence, there are quantum permutation groups that cannot be realised as the quantum automorphism groups of a complete edge and vertex labelled graph.
\end{remark}

\section{Coherent and transitive correlations}\label{sec:coherent-transitive}

All correlations considered in this section have the same questions and answers set, viz., $W$. A correlation $p = (p(a,b|x,y))$ is called \emph{transitive} if for all $a,b,x,y\in W$, if we have $p(a,b|x,y) = 0$, then $p(a,b|r,s)p(r,s|x,y) = 0$ for all $r,s\in W$. 

A correlation $p = (p(a,b|x,y))$ induces a (nonempty) relation $\cl R$ on $W\times W$ by: for $(a,b), (x,y) \in W\times W$ we say $(a,b) \sim (x,y)$ if $p(a,b|x,y) \neq 0$. It is again straightforward to verify that if $p$ is transitive, then its induced relation $\cl R$ is transitive. 

A correlation $p$ is said to be \emph{coherent} if $p$ is bisynchronous and its induced relation $\cl R$ on $W\times W$ is an equivalence relation. (In fact, we only need the relation to be symmetric and transitive: for each $(x,y)\in W\times W$, we have $\sum_{a,b}p(a,b|x,y) = 1$ which implies that there exists some $(a,b)\in W\times W$ such that $(a,b)\sim (x,y)$. By symmetry, $(x,y)\sim (a,b)$, and hence by transitivity, $(x,y)\sim (x,y)$, as needed for reflexivity.) This nomenclature is partly justified by the following lemma.

\begin{lemma}
Let $p\in C_{qc}^{bs}$ be a coherent correlation and let $\cl R$ be the associated relation on $W\times W$. Then the partition of $W\times W$ corresponding to the equivalence relation $\cl R$ is a coherent configuration.
\end{lemma}

\begin{proof}
We verify the three conditions (a)--(c) in \cref{def:coherent-config}. Let the equivalence classes arising from the equivalence relation $\cl R$ be given by $\{R_i:i\in \cl I\}$.

To verify property (a), let $x\in W$. Since $\cl R$ is a partition of $W\times W$, there exists some $R_i\in \cl R$ such that $(x,x)\in R_i$. It then suffices to show that $R_i$ does not contain ``off-diagonal'' elements. Let $(a,b)\sim (x,x)$. Then $p(a,b|x,x) \neq 0$, and since $p$ is synchronous, we must have $a = b$.

To show (b), fix a class $R_i$, and let $(a,b),(x,y)\in R_i$, that is, $p(a,b|x,y) \neq 0$. By \cite[Theorem 5.5]{PSSTW}, $p(b,a|y,x) = p(a,b|x,y) \neq 0$, so that $(b,a)\sim (y,x)$. Thus the converse of $R_i$ is again an equivalence class.

To show (c), we borrow an argument used in the proof of \cite[Theorem 4.6]{lupini}. By \cite[Theorem 5.5]{PSSTW}, there exists a $C^*$-algebra 7 $\cl A$ equipped with a faithful tracial state $\tau$ and having projections $\{u_{x,a}:x,a\in W\}$ where $\sum_{a\in W} u_{x,a} = 1$ for all $x\in W$ such that $p(a,b|x,y) = \tau(u_{x,a}u_{y,b})$ for all $a,b,x,y\in W$. Fix $i,j,k\in \cl I$. Let $(x,z), (x^{\prime},z^{\prime})\in R_k$, and define \begin{align*}
S &= \{y\in W: (x,y) \in R_i \text{ and } (y,z) \in R_j\}, \\ S^{\prime} &= \{y^{\prime}\in W: (x^{\prime},y^{\prime}) \in R_i \text{ and } (y^{\prime},z^{\prime}) \in R_j\}.
\end{align*} We want to show that $|S| = |S^{\prime}|$. Since $(x,z)\sim (x^{\prime},z^{\prime})$ we see that $u_{x,x^{\prime}}u_{z,z^{\prime}} \neq 0$. Then, $|S| = |S^{\prime}|$ follows from the following sequence of equalities: \begin{align*}
|S|u_{x,x'}u_{z,z'} & = u_{x,x'} \left( \sum_{y \in S} 1 \right) u_{z,z'} = u_{x,x'} \left(\sum_{y \in S} \sum_{y'\in W} u_{y,y'}\right) u_{z,z'} \\
&= \sum_{y \in S}\sum_{y'\in W} u_{x,x'}u_{y,y'}u_{z,z'} \\
&= \sum_{y \in S}\sum_{y' \in S'} u_{x,x'} u_{y,y'} u_{z,z'} = \sum_{y' \in S'}\sum_{y \in S} u_{x,x'} u_{y,y'} u_{z,z'} \\
&= \sum_{y' \in S'}\sum_{y\in W} u_{x,x'} u_{y,y'} u_{z,z'} = u_{x,x'} \left(\sum_{y' \in S'} \sum_{y \in W} u_{y,y'}\right) u_{z,z'} = |S'| u_{x,x'} u_{z,z'}.
\end{align*} Here, in the fourth equality, we have used the following fact: if $y\in S$ and $y^{\prime}\notin S^{\prime}$, then $u_{x,x^{\prime}}u_{y,y^{\prime}}u_{z,z^{\prime}} = 0$.  To see this, observe that $y^{\prime}\notin S^{\prime}$ if and only if either $(x^{\prime},y^{\prime})\notin R_i$ or $(y^{\prime},z^{\prime})\notin R_j$. If $(x^{\prime},y^{\prime})\notin R_i$, then since $y\in S$ we have $(x,y)\in R_i$. Thus $(x^{\prime},y^{\prime})\nsim (x,y)$, so that $\tau(u_{x,x^{\prime}}u_{y,y^{\prime}}) =  p(x^{\prime},y^{\prime}|x,y) = 0$, and so by faithfulness of the state $\tau$ we have $u_{x,x^{\prime}}u_{y,y^{\prime}} = 0$. Similarly, if $(y^{\prime},z^{\prime})\notin R_j$, then since $(y,z)\in R_j$, we have $(y^{\prime},z^{\prime})\nsim (y,z)$, and thus $u_{y,y^{\prime}}u_{z,z^{\prime}} = 0$. Either way, $u_{x,x^{\prime}}u_{y,y^{\prime}}u_{z,z^{\prime}} = 0$. A similar argument is invoked for the sixth equality. Finally, the last equality follows from bisynchronicity of the correlation $p$ \cite[Remark 2.1]{paulsen}.
\end{proof}

\subsection{Correlation generated by the fundamental representation and Haar state}
We begin with a simple family of coherent correlations, and give necessary and sufficient conditions for them to be nonlocal --- they are nonlocal if and only if the coherent configuration associated with them is non-Schurian.

Let $\mathbb{G}$ be a quantum permutation group acting on a set $W$ with fundamental representation $[u_{i,j}]_{i,j\in W}$, and Haar state $h$. In \cite[Theorem 4.2]{lupini}, it was shown that if $R_1,\dots,R_m\subseteq W\times W$ are the quantum orbitals of $\mathbb G$, then
\begin{align*}
h(u_{x,a}u_{y,b}) = 
    \begin{cases}
        \frac{1}{|R_i|} & \text{ \ if \ } (x,y),\ (a,b) \in R_i, \\
        0 & \text{ \ if \ } (x,y) \nsim_2 \ (a,b).
    \end{cases}
\end{align*} Since the Haar state is tracial and $[u_{i,j}]_{i,j\in W}$ is a magic unitary, we see that defining $p_{\mathbb{G}}(a,b|x,y) :=  h(u_{x,a}u_{y,b})$  (for all $a,b,x,y\in W$) yields a quantum commuting correlation. We note that the correlation $p_{\mathbb{G}}$ only depends on the quantum orbitals of $\mathbb{G}$. From \cite[Theorem 3.10]{lupini}, we know that the quantum orbitals of a quantum permutation group form a coherent configuration. Hence, the correlation $p_{\mathbb G}$ is a coherent correlation.

\begin{prop}
Let $\mathbb{G}$ be a quantum permutation group acting on a set $W$. Then, $p_{\mathbb{G}}$ is local if and only if the coherent configuration of quantum orbitals of $\mathbb{G}$ is Schurian.
\end{prop}
\begin{proof}

First, let us assume that the quantum orbitals of $\mathbb{G}$ form a Schurian coherent configuration. Then, there is a finite group $G \subseteq S_W$, whose orbitals are the same as the quantum orbitals of $\mathbb{G}$. It is then easy to see that $C(G)$, its Haar state and its fundamental representation induce the same correlation $p_{\mathbb G}$ as the Haar state and the fundamental representation of $\mathbb{G}$. Since $C(G)$ is a commutative $C^*$-algebra, it follows that $p_{\mathbb G}$ is indeed local.

Now, let us assume that $p_{\mathbb G}$ is local. Then, it follows from \cref{t21} that there is a commutative $C^*$-algebra $\cl A$, with projections $\{p_{i,j}:i,j\in W\}$, satisfying $\sum_{i\in W} p_{i,j} = \sum_{j\in W} p_{i,j} = 1$ and a faithful tracial state $s$, such that $p_{\mathbb G}(a,b|x,y) = s(p_{x,a}p_{y,b})$. In particular, the universal $C^*$-algebra generated by $\{u_{i,j}:i,j\in W\}$ satisfying: \begin{enumerate}
    \item[(a)] $u_{i,j}^2 = u_{i,j}^* = u_{i,j}$,
    \item[(b)] $\sum_{i\in W} u_{i,j} = \sum_{j\in W} u_{i,j} = 1$,
    \item[(c)] $p(a,b|x,y) = 0 \Rightarrow u_{x,a}u_{y,b} = 0$,
    \item[(d)] $u_{x,a}u_{y,b} = u_{y,b}u_{x,a}$,
\end{enumerate} for all $i,j,x,y,a,b$ is non-trivial. Since $[u_{i,j}]$ is a magic unitary, and since this universal $C^*$-algebra is commutative, it is of the form $C(G)$, where $G \subseteq S_W$ is a permutation group. It is easy to see that the orbitals of $G$ are the same as that of $\mathbb{G}$. Hence, the coherent configuration of the quantum orbitals of $\mathbb{G}$ is indeed Schurian by construction.
\end{proof}

To give an explicit example of a quantum permutation group whose quantum orbitals form a non-Schurian coherent configuration, we first need the following result.

\begin{theorem}[{\cite[Theorem 4.5]{lupini}}]\label{thm:qc-isomorphic-iff}
    Let $X$ and $Y$ be connected graphs. Then, $X \cong_{qc} Y$ if and only if there exists $x\in V(X)$ and $y \in V(Y)$ that are in the same orbit of $\mathrm{Qut}(X \sqcup Y)$.
\end{theorem}

\begin{prop}
    Let $X$ and $Y$ be two connected non-isomorphic graphs such that $X \cong_{qc} Y$. Then, the quantum orbitals of $\mathrm{Qut}(X \sqcup Y)$ form a non-Schurian coherent configuration. Hence, $p_{{\mathrm{Qut}{(X \sqcup Y)}}}$ is nonlocal.
\end{prop}

\begin{proof}
Let $[u_{ij}]$ be the fundamental representation of $\mathrm{Qut}(X \sqcup Y)$ and consider the quantum orbitals of $\mathrm{Qut}(X \sqcup Y)$. From \cref{thm:qc-isomorphic-iff}, it follows that there exists $x \in V(X)$ and $y \in V(Y)$ such that $u_{x,y} \neq 0$. Now, let us assume that there exists some permutation group $G \subseteq S_{V(X) \sqcup V(Y)}$ which induces the same orbitals as $\mathrm{Qut}(X \sqcup Y)$. Let $[v_{ij}]$  be the fundamental representation of $C(G)$. We see that if $\mathrm{rel}(i,k) \neq \mathrm{rel}(j,l)$, then $v_{i,j}v_{k,l} = 0$. Moreover, the $v_{i,j}$'s commute. Hence, $G \subseteq \mathrm{Aut}(X \sqcup Y)$. Since $v_{x,y} \neq 0$, we see that there exists a graph automorphism $g$ of $X \sqcup Y$ mapping $x$ to $y$. This, however, is impossible since $X$ and $Y$ are not isomorphic. Hence, the coherent configuration induced by the quantum orbitals of $\mathrm{Qut}(X \sqcup Y)$ is non-Schurian. \end{proof}

\subsection{General coherent correlations}
Now that we have established a strong connection between the nonlocality of a coherent correlation and the coherent configuration associated with it, we try to generalise similar results for all coherent correlations. 

Let $\Omega = \{R_i\}_{i \in I}$ be a coherent configuration. Then, it is easy to verify that the set of all coherent correlations $p$, for which $\sim_p$ induces the same coherent configuration as $\Omega$ is a semigroup. We shall denote this semigroup by $P_{\Omega}$. By \cref{l37}, the hardest game $\mathscr{G}_{\Omega}$ that can be won using these strategies $P_{\Omega}$ is a bisynchronous transitive game. Let us denote the equivalence relation by $\sim_{\Omega}$. The predicate $V_{\Omega}$ of this game $\mathscr{G}_{\Omega}$ is given by \begin{align*}
    V_{\Omega}(a,b|x,y) =
    \begin{cases}
        0 & \text{if \ } (x,y) \nsim_{\Omega} (a,b), \\
        1 & \text{if \ } (x,y) \sim_{\Omega} (a,b).
    \end{cases}
\end{align*} Since $\mathscr{G}_{\Omega}$ is a transitive game, the set of all perfect qc-strategies for $\mathscr{G}_{\Omega}$ is also a semigroup. If $P_{\Omega}$ is non-empty, $\mathscr{G}_{\Omega}$ has a non-trivial $C^*$-algebra $C^*(\mathscr{G}_{\Omega})$, which is a quantum permutation group. In particular, since all the correlations in $P_{\Omega}$ are perfect qc-strategies for $\mathscr{G}_{\Omega}$, the following results hold.

\begin{prop}
    Let $p$ be a coherent qc correlation. Then, there exists
    \begin{enumerate}[label = \roman*.]
        \item a quantum permutation group $\mathbb{G}$ with fundamental representation $[u_{i,j}]$ acting on a set  $W$, whose orbitals are same as the coherent configuration induced by $\sim_p$,
        \item a $C^*$-algebra $\cl A$ equipped with a faithful tracial state $s$, projections $e_{i,j} \in \cl A$ for $i,j \in W$, such that $\sum_{j} e_{i,j} = \sum_{i} e_{i,j} = 1$ and $V(a,b|x,y) = 0 \Rightarrow e_{x,a}e_{y,b} = 0$,
        \item and a $*$-homomorphism from $C(\mathbb{G}) \to \cl A$ taking $u_{i,j} \to e_{i,j}$ 
    \end{enumerate}
    such that $p(a,b|x,y) = s(e_{x,a}e_{y,b})$
\end{prop}

The following corollary is merely a restatement of the previous result.

\begin{cor}
    Let $p$ be a coherent qc correlation. Then, there exists a quantum permutation group $\mathbb{G} \subseteq S_W^+$, with fundamental representation $[u_{i,j}]$ acting on $W$, whose orbitals are the same as the partitions induced by $\sim_p$ and a tracial state $s$ on $C^*(G)$, such that $p(a,b|x,y) = s(u_{x,a}u_{y,b})$.
\end{cor}
\begin{proof}
    Let $\pi$ be the $*$-homomorphism from $C^*(G)$ to $A$ defined by $u_{i,j} \mapsto e_{i,j}$. It is now easy to verify that $s \circ \pi$ is a tracial state on $C^*(\mathscr{G}_{\Omega_p})$ satisfying the required conditions.
\end{proof}

\begin{theorem}\label{thm:coh-corr-prop-dist}
    Let $p$ be a coherent local correlation. Then, there exists a group $G \subseteq S_n$ acting on $W$, whose orbitals are the same as the partitions induced by $\sim_p$, and a probability distribution $\pi$ on $G$ such that \begin{align*}
        p(a,b|x,y) = \sum_{\substack{\{g \in G\colon g(x) = a \\ \& \ g(y) = b\}}} \pi(g).
    \end{align*}
\end{theorem}

\begin{proof}
    Since $p$ is a coherent correlation, we see that there exists a state $s$ on $C^*(\mathscr{G}_{\Omega_p})$ such that $p(a,b|x,y) = s(u_{x,a}u_{y,b})$. Let $\widetilde{C^*}(G_{\Omega_p})$ be the quotient of $C^*(G)$ by the relations $u_{x,a}u_{y,b} = u_{y,b}u_{x,a}$ for all $a,b,x,y\in W$. Since $p$ is a local correlation, $\widetilde{C^*}(\mathscr{G}_{\Omega_p})$ is nontrivial. Moreover, $\widetilde{C^*}(\mathscr{G}_{\Omega_p})$ is a commutative $C^*$-algebra, and also a quantum permutation group. Hence, there exists a group $G \subseteq S_W$, such that $\widetilde{C^*}(\mathscr{G}_{\Omega_p}) = C(G)$. It is easy to see that the orbitals of $G$ are the same as the partitions induced by $\sim_p$. 

    Since $p$ is a local correlation, there exists a commutative $C^*$-algebra $A$ equipped a faithful tracial state $s$, projections $e_{i,j} \in A$ for $i,j \in W$, such that $\sum_{j} e_{i,j} = \sum_{i} e_{i,j} = 1$, and a $*$-homomorphism $\phi$ from $C^*(G) \to A$ taking $ u_{i,j} \to e_{i,j}$, where $[u_{i,j}]_{i,j}$ is the fundamental representation of $\mathbb{G}$, such that $p(a,b|x,y) = s(e_{x,a}e_{y,b})$. Hence, $u_{i,j} \to e_{i,j}$, defines a $*$-homomorphism $\widetilde{\phi}$, which implies that $s \circ \widetilde{\phi}$ is a state on $C(G)$. We know that states on $C(G)$ correspond to regular complete probability measures on $G$. Since $G$ is finite, the state $s\circ \widetilde{\phi}$ corresponds to a probability distribution $\pi$ on $G$. 

    Moreover, we have that \begin{align*}
        p(a,b|x,y)  = s(u_{x,a}u_{y,b}) = \int_G u_{x,a}u_{y,b} d\pi  = \sum_{\substack{g \in G\colon g(x) = a \\ \& \ g(y) = b}} \pi(g),
    \end{align*} as needed.
\end{proof}

As a corollary we get:

\begin{cor}
    Let $p$ be a coherent quantum commuting correlation. If the coherent configuration of partitions induced by $\sim_p$ is non-Schurian, $p$ is nonlocal.
\end{cor}

\begin{proof} 
    Let us assume that $p$ is a local coherent correlation such that the partitions induced by $\sim_p$ form a non-Schurian coherent configuration. From \cref{thm:coh-corr-prop-dist} it follows that there exists a group $G \subseteq S_W$, whose orbitals are the same as the partitions induced by $\sim_p$. This proves the contrapositive of the statement that was to be proved.
\end{proof}

We should note that when the coherent configuration of partitions induced by $\sim_p$ is Schurian, the correlation can either be local or nonlocal.

\begin{example}
    Let us consider a classical group $G \subseteq S_n$. Then, the correlation induced by the Haar state of $G$, and its fundamental representation $p_G$ is local. It is easy to see that $\sim_{p_G}$ is an equivalence relation, and partitions induced by it are precisely the orbitals of $G$, which is a coherent configuration.
\end{example}

\begin{example}
    Let us take any nonlocal correlation $q$ and consider $\lambda q + (1-\lambda) p_{S_n}$ for some $0 < \lambda < 1$, where $p_{S_n}$, is the correlation induced by the Haar state and the standard representation of $S_n$. Then, clearly $p(a,b|x,y) = 0$ if and only if $\delta_{x,y} \neq \delta_{a,b}$. Hence, $p$ is nonlocal and $\sim_p$ partitions $[n] \times [n]$ into the same partitions as the orbitals of $S_n$.  
\end{example}

\subsection{Transitive correlations}

Now, we want to establish a similar set of results for transitive correlations. Proceeding in the same fashion as the case of coherent correlations, we can prove the following proposition.

\begin{prop}
    Let $\sim$ be a transitive relation on $W \times W$ such that
    \begin{enumerate}[label = \roman*.]
        \item $(a,b)\not \sim (x, x)$ if $a\neq b$ 
        \item $(a,b) \sim (x,y)$ if and only if $(b,a) \sim (y,x)$
    \end{enumerate}
    Then, the set $P_\sim$ of all transitive correlations $p$ such that $\sim \ = \ \sim_p$, if it is non-empty, forms a semigroup.
\end{prop}

Given such a transitive relation, it is easy to see that the hardest nonlocal game that can be won using these strategies is a synchronous transitive game, say $\mathscr{G}_{\sim}$. Since $\mathscr{G}_{\sim}$ is a synchronous transitive game, the $C^*$-algebra of the game $C^*(\mathscr{G}_{\sim})$ is a compact quantum semigroup. Hence, proceeding in a similar manner as we did in the case of coherent correlations, the following results follow.

\begin{prop}
    Let $p$ be a transitive qc correlation. Then, there exists a compact quantum semigroup $\mathbb{G}$ acting on $W$, generated by projections $u_{i,j}$ such that $\sum_{j}u_{i,j} = 1$ and a tracial state $s$ on $C(\mathbb{G})$ such that
    such that $p(a,b|x,y) = s(u_{x,a}u_{y,b})$.
\end{prop}
\begin{proof}
    Follows from the fact that $C^*(\mathscr{G}_{\sim_p})$ is a compact quantum semigroup.
\end{proof}

\begin{theorem}
    Let $p$ be a local transitive correlation. Then, there exists a semigroup $G$ of functions from $W$ to $W$, and a probability distribution $\pi$ on $G$ such that \begin{equation*}
        p(a,b|x,y) = \sum_{\substack{f \in G: f(x) = a \\ \& \ f(y) = b}} \pi(f).
    \end{equation*}
\end{theorem}
\begin{proof}
    Let us take $\widetilde{C^*}(G_{\sim_p})$ to be the universal $C^*$-algebra generated by projections $u_{i,j}$ such that 
    \begin{enumerate}
        \item[(a)] $\sum_j u_{i,j} = 1$
        \item[(b)] $u_{x,a}u_{y,b} = 0$ if $p(a,b|x,y) = 0$ 
        \item[(c)] $u_{x,a}u_{y,b} = u_{y,b}u_{x,a}$
    \end{enumerate}
    Then, $\widetilde{C^*}(\mathscr{G}_{\sim_p})$ is a compact quantum semigroup acting on $W$. Since $p$ is a local transitive correlation, $\widetilde{C^*}(\mathscr{G}_{\sim_p})$ is non-trivial. Hence, there exists a semigroup $G$ such that $\widetilde{C^*}(\mathscr{G}_{\sim_p}) = C(G)$. It is easy to see that $G$ is a semigroup of functions from $W$ to $W$, since the action of $G$ on $W$ can be constructed from the action of $\widetilde{C^*}(\mathscr{G}_{\sim_p})$ on $\mathbb{C}^W$. Since $p$ is a perfect local strategy of $\mathscr{G}_{\sim_p}$, it corresponds to a state $s$ on $C(G)$, i.e., a probability distribution $\pi$ on $G$ such that \begin{align*}
        p(a,b|x,y) = s(u_{x,a}u_{yb}) = \int_G u_{x,a}u_{y,b} d\pi \quad = \sum_{\substack{f \in G: f(x) = a \\ \& \ f(y) = b}} \pi(f),
    \end{align*} as needed.
\end{proof}

\paragraph*{\textbf{\emph{Acknowledgements}}}
Authors PK and DR are supported by the Carlsberg Foundation Young Researcher Fellowship CF21-0682 -- ``Quantum Graph Theory". JP was partially supported by NSF Award No. 2210399 and AFOSR Award No. FA9550-20-1-0067.

\bibliography{bibfile.bib}
\end{document}